\documentclass[sigconf]{acmart}
\pdfoutput=1

\usepackage{booktabs} 
\usepackage{nccmath}
\usepackage{algorithmic}
\usepackage{graphicx}
\usepackage{textcomp}
\usepackage{xcolor}
\usepackage{wrapfig}
\usepackage{tabularx}
\usepackage{multirow}
\usepackage{array}
\usepackage{standalone}
\usepackage{mathrsfs}  
\usepackage{booktabs}
\usepackage{marvosym}
\usepackage{stfloats}
\usepackage{stmaryrd}
\usepackage{mathtools}

\usepackage{todonotes}
\setuptodonotes{inline}

\usepackage{bussproofs}
\usepackage{tikz-uml}
\usetikzlibrary{backgrounds,decorations.pathreplacing, calc, fit, shapes.arrows, positioning}

\usepackage{csquotes}
\usepackage{graphicx}
\usepackage{paralist}
\usepackage{rotating}
\usepackage{xspace}
\usepackage{xparse}
\usepackage{hhline}

\usepackage{caption}
\usepackage{subcaption}

\usepackage{microtype}

\def\+{\discretionary{}{}{}}

\usepackage{hyperref}
\usepackage[nameinlink,capitalize]{cleveref}
\Crefname{section}{Sec.}{Sects.}
\Crefname{equation}{Eq.}{Eqs.}
\Crefname{figure}{Fig.}{Figs.}
\Crefname{tabular}{Tab.}{Tabs.}
\Crefname{definition}{Def.}{Defs.}
\Crefname{theorem}{Thm.}{Thms.}

\newcommand\pre{\mathit{pre}}
\newcommand\post{\mathit{post}}
\renewcommand{\sec}{\mathrm{s}}

\usepackage{listings}
\NewDocumentCommand{\lstcenter}{v}%
{\begin{center}
\begin{tabular}{c}
\ttfamily #1 
\end{tabular}
\end{center}}

\newcommand\KeY{Ke\kern-0.75ptY\xspace}
\newcommand\QAC{\textrm{Qu}\textsc{ac}\xspace}

\crefname{lstlisting}{listing}{listings}
\Crefname{lstlisting}{Listing}{Listings}

\newcommand{\awc}[2][]{\footnote{\awadd{#2}}}

\newcommand{\awadd}[2][]{\begingroup\color{purple} #2\endgroup}
\newcommand{\awtodo}[1]{\textcolor{purple!70!black}{#1}}
\newcommand{\fltodo}[1]{\todo[color=orange]{#1}}
\newcommand{\frtodo}[1]{\textcolor{violet!50}{#1}}
\newcommand{\sttodo}[1]{\todo[color=green!50!black]{#1}}
\newcommand{\cmtodo}[1]{\textcolor{teal}{#1}}

\let\oldparagraph=\paragraph
\renewcommand\paragraph[1]{\oldparagraph{#1.}}

\copyrightyear{2024}
\acmYear{2024}
\setcopyright{rightsretained}
\acmConference[SAC '24]{The 39th ACM/SIGAPP Symposium on Applied Computing}{April 8--12, 2024}{Avila, Spain}
\acmBooktitle{The 39th ACM/SIGAPP Symposium on Applied Computing (SAC '24), April 8--12, 2024, Avila, Spain}\acmDOI{10.1145/3605098.3636008}
\acmISBN{979-8-4007-0243-3/24/04}

\acmDOI{10.1145/3605098.3636008}

\begin{document}

\title{Quantifying Software Correctness by Combining Architecture Modeling and Formal
  Program Analysis}

\renewcommand{\shorttitle}{Quantifying Software Correctness}

\author{Florian Lanzinger}
\affiliation{
\institution{Karlsruhe Institute of Technology}
\streetaddress{Am Fasanengarten 5}
\city{Karlsruhe}
\country{Germany}
}
\orcid{0000-0001-8560-6324}
\email{lanzinger@kit.edu}

\author{Christian Martin}
\affiliation{
\institution{Karlsruhe Institute of Technology}
\streetaddress{Am Fasanengarten 5}
\city{Karlsruhe}
\country{Germany}
}
\orcid{0009-0004-4332-1194}
\email{christian.martin@kit.edu}

\author{Frederik Reiche}
\affiliation{
\institution{Karlsruhe Institute of Technology}
\streetaddress{Am Fasanengarten 5}
\city{Karlsruhe}
\country{Germany}
}
\orcid{0000-0002-5993-0558}
\email{frederik.reiche@kit.edu}

\author{Samuel Teuber}
\affiliation{
\institution{Karlsruhe Institute of Technology}
\streetaddress{Am Fasanengarten 5}
\city{Karlsruhe}
\country{Germany}
}
\orcid{0000-0001-7945-9110}
\email{teuber@kit.edu}

\author{Robert Heinrich}
\affiliation{
\institution{Karlsruhe Institute of Technology}
\streetaddress{Am Fasanengarten 5}
\city{Karlsruhe}
\country{Germany}
}
\orcid{0000-0003-0779-9444}
\email{robert.heinrich@kit.edu}

\author{Alexander Weigl}
\affiliation{
\institution{Karlsruhe Institute of Technology}
\streetaddress{Am Fasanengarten 5}
\city{Karlsruhe}
\country{Germany}
}
\orcid{0000-0001-8446-4598}
\email{weigl@kit.edu}

\renewcommand{\shortauthors}{Lanzinger et al.}

\begin{abstract}
  Most formal methods see the correctness of a software system as a binary decision.
  However, proving the correctness of complex systems completely is difficult
  because they are composed of multiple components, usage scenarios, and environments. 
  We present \QAC, a modular approach for quantifying the correctness of service-oriented software systems by combining software architecture modeling with deductive verification.
  Our approach is based on a model of the service-oriented architecture and 
  the probabilistic usage scenarios of the system.
  The correctness of a single service is approximated by a coverage region, which is a formula describing which inputs for that service are proven to not lead to an erroneous execution.
  The coverage regions can be determined by a combination of various analyses, e.g., formal verification, expert estimations, or testing.
  The coverage regions and the software model are then combined into
  a probabilistic program. From this, we can compute the probability that under a given usage profile no service is called outside its coverage region.
  %
  %
  We also present an implementation of \QAC{} for Java using the modeling tool Palladio and the deductive verification tool \KeY.
  We demonstrate its usability by applying it to a software simulation of an energy system.
\end{abstract}



\maketitle

\section{Introduction}
\label{sec:intro}

\paragraph{Motivation}
Vehicles and critical infrastructure are increasingly governed by software. Therefore, verifying that software behaves according to its specification is becoming ever more important. However,
most formal methods used to prove a program's correctness simply output a binary
decision, which does not do justice to large software systems which
combine many different components in a complex usage environment.

An individual component may behave provably correctly under assumptions that are not ensured by its environment or, conversely, a component may behave incorrectly for inputs that never appear when used within the evaluated system.
Both cases can also be considered more gradually: It may be the case that inputs leading to incorrect behavior are very rare or extremely frequent.
This demonstrates the necessity for analyses which reach across multiple components and for more gradual, quantitative assessments.
There are some non-binary measures of software reliability based on source code in the literature like test coverage.
However, such approaches often do not provide guarantees on the system's reliability for concrete usage scenarios.
Many quantitative source code analyses only consider individual components, in which case it remains unclear how to combine the metrics.
Related work in the architectural domain, such as by Brosch et al.~\cite{brosch-12}, computes reliability measures based on concrete usage scenarios, the system structure, and abstract behavioral specifications.
However, the reliability-relevant information is only estimated by an expert.

Our work proposes to couple a quantitative analysis of the architecture model with a formal analysis of source code.
By combining these analyses, we can reason about larger systems while still providing fine-grained, quantitative feedback on reliability.
In addition, the coupling enables the usage of implementation details in the architectural analysis rather than relying on assumptions by an software architect~\cite{SchulzReicheHahner2022_1000143320}.

\paragraph{Contribution}

We introduce \QAC{} (``\emph{Qu}antifying \emph{A}rchitecture and \emph{C}ode''), an approach to quantitatively determine the probability that a piece 
of software respects its contract by combining architecture modeling and formal source-code analysis.
\QAC{} answers the following query: \emph{``With
what probability does a typical usage of the system not lead to an error?''}
By \emph{error}, we mean any program state contrary to a method's contract.
Specifically, an error occurs whenever a method terminates in a state violating its postcondition and whenever it calls another method in a state violating that method's precondition.
To describe these states, we use \emph{coverage regions} for every method: a coverage region is a formula describing inputs for which the current method is proven to behave correctly.
All coverage regions are inserted into an architectural model. Based on this model and a usage profile describing the typical system usage, we use model counting to calculate a \emph{coverage probability}, which is an under-approximation of
the \emph{correctness probability} that the overall program behaves correctly.

To find coverage regions, we consider several approaches, which makes \QAC{} applicable both in the early and late stages of software development.
In the early stages, they are estimated by
the developers. These estimations are later replaced by testing and verification
results, which are the main approaches considered in this paper.
In operation, the usage profile and coverage regions can be further updated by information monitored at run time.
This makes \QAC{}
incremental: Information gathered in different phases leads to more precise modeling and thus to a better risk assessment.
In addition, it allows developers to combine static verification approaches, which are powerful but expensive -- especially when aiming for 100\% proof coverage -- with run-time verification and testing approaches in a rigorous way.

We formally introduce the \QAC{} approach and show that
it is sound, i.e., it never overestimates the
correctness probability.
We provide a concrete implementation using the Palladio Component Model (PCM)~\cite{PalladioBook2016},
a metamodel for component-based software architectures and usage profiles,
\KeY~\cite{KeYBook2016}, an interactive theorem prover, and several probabilistic model counters.
Finally, we demonstrate its usability in a case study.

\paragraph{Limitations}
\QAC{} can be applied to programs modeled as a service-oriented, component-based architecture.
We require an association between source-code elements (methods and classes) and architectural elements (services and components).
In addition, the logic in which the coverage regions are expressed is restricted by the analysis tools (e.g., \KeY) and the model counter. 
So far, \QAC{} can be used for safety properties, i.e., the probability of certain errors given certain input distributions. It is not suitable for other properties like liveness.
The implementation
only supports synchronous calls and the execution of a single usage profile. It does not support multi-agent analyses.
It is also limited to programs without unbounded loops or recursion. We plan to extend it to allow support for more general
control flow structures both in the service models and in the extraction
of coverage regions.

\section{Preliminaries}
\label{sec:prelim}

\paragraph{Architecture and code}
We combine fine-grained analyses on the source-code level with coarse-grained analyses on the architectural level.
On the source-code level, we have an object-oriented program consisting of \emph{classes} which provide \emph{methods}. On the architectural level, we have a service-oriented, component-based architecture consisting of \emph{components} which provide \emph{services}.
We assume a mapping that maps every component to a class and every service to a method (not necessarily vice versa).
Component-based behavioral specifications provide, in general, no information about the inner state of a component, i.e., state variables. 
Every service or method also has a contract consisting of two formulas $(\pre,\post)$, where $\pre$ is the precondition, which must hold whenever the service is called, and $\post$ is the postcondition, which must hold whenever the service terminates.
In addition to the structure of the software, we also have \emph{behavioral specifications}, which are rough approximations of the internal
behavior of a given service.
Such specifications include, for example, control flow constructs, like branches or loops, and calls to other services.
We assume that every behavioral specification in the architecture is respected by the implementation of the corresponding method.
This consistency can be achieved manually, by applying consistency preservation approaches like the one by Monschein et al.~\cite{Monschein21} or Vitruvius~\cite{KLARE2021110815}, or by inferring the behavioral specifications from the source code.

\paragraph{Software modeling with the Palladio Component Model}
We build upon service-oriented, component-based architectures,
which can, e.g., be modeled in the Palladio Component Model
(PCM)~\cite{PalladioBook2016}. In the PCM, architectures consist of independent
components, which provide services to the user or to other components.
A component must declare all services from other components that it requires.
Two components are connected if one component provides a service required by the other one.
Therefore, a service can be called from outside a component; it takes a specified number of parameters and returns a result. Service calls are always synchronous.
For every service, a behavioral specification is provided in form of Service Effect Specifications (SEFFs).
A SEFF contains, among others, nodes for control structures like loops and branches, the usage of parameters in the behavior, external call actions modeling calls to other services and the description of results from external calls and their processing.
Conditions in control-flow structures are defined by expressions over values provided by local variables, the parameters and other arbitrary specifications.
This SEFF is an abstraction of the source code that actually implements the service: It exactly specifies the service's behavior w.r.t.\ calls to other services and modifications of the component state, but the implementation may contain additional computations or optimizations not represented in the SEFF.
We also use the PCM's usage profile, which describes the usage scenario of the
system in form of a flow chart and calls of the public services. It also specifies probabilistic values to model the
unknown user input. In essence, usage profiles capture the observation of how the system is or should be used.
These probabilistic values are provided either from an expert of the systems domain from prior projects (educated guess) or by measurements~\cite{PalladioBook2016}. 
Measurements can be taken from prior versions of the product in an evolution scenario or in early versions taken from a test group.
The described system behavior under a usage profile is a computation tree of
service calls assembled by the execution of the usage profile and the directly or
indirectly called SEFFs and the applied control structures. Because of the probabilistic usage profile, each
computation path in the tree has an associated  probability. 

\paragraph{Source code analysis}
The most important tool we use to analyze the source code is \KeY, a deductive verification tool used to prove Java programs specified in the Java Modeling Language  (JML)~\cite{JMLrefman}. \KeY{} is semi-automatic, meaning that most JML specifications can be proven automatically, but for more complex cases, the proof can be manually inspected and guided by a human verifier.

We define the necessary concepts of source code analysis through the lens of dynamic logic (DL).
A detailed overview of DL is given by Harel et al.~\cite{Harel79,harel2001dynamic}.
Generic first-order DL is a multi-modal logic which extends first-order logic (FOL) with programs that describe possible state transitions.
To assert that some postcondition $\post$ holds after execution of a program $\alpha$, we use the box modality: $\left[\alpha\right]\post$.
Using the usual logic operators from FOL, we can then specify contracts, such as the following, which asserts that, if the variable $x$ starts with the value $42$, then it will always be larger than $42$ after we run the program in brackets for any positive $a$:
$
x = 42 \land a > 0 \rightarrow \left[ x \coloneqq x+a \right] x > 42
$.
Formulas in DL are evaluated using Kripke structures with state transitions where each state $\sigma \in \mathcal{S}$ contains (among other things) a first-order logic structure assigning each variable a value of its respective domain.
We use $\sigma \vDash \rho$ to denote that a DL formula $\rho$ holds in $\sigma\in\mathcal{S}$ and $\vDash \rho$ to denote that a formula is valid, i.e., holds in all states.

There are many instances of DL, including Java Dynamic Logic (JavaDL)~\cite[Ch. 3]{KeYBook2016}, which is implemented in \KeY{} to deductively verify Java programs specified by JML contracts.
\KeY{}'s calculus operates on \emph{sequents} of the form $\phi \Longrightarrow \psi$, where $\phi$ is the \emph{antecedent} and consists of $n$ formulas $\phi_i$ while $\psi$ is the \emph{succedent} and consists of $m$ formulas $\psi_j$.
A sequent is satisfied by a state $\sigma \in \mathcal{S}$ iff $\sigma \vDash \bigwedge_{i} \phi_i \rightarrow \bigvee_{j} \psi_j$.
A sequent satisfied by all states is called valid.
A formula $\rho$ can be proven true by applying axioms and rules that construct a proof tree: The root is the sequent $\Longrightarrow \rho$, where the succedent contains the formula to prove and the antecedent is empty.
By applying a rule to a node $C$, we obtain several child nodes $P_1,\ldots,P_n$ such that the validity of all sequents $P_i$ together implies the validity of the sequent $C$.
Rules may be \emph{locally} or \emph{globally sound}~\cite{KeYCoverage}:
A rule is \emph{globally sound} if the validity of all $P_i$ implies the validity of $C$ and \emph{locally sound} if every state satisfying all $P_i$ also satisfies $C$.
The tree's leaves are called goals and may be closed -- i.e., tautologies -- or open -- i.e., yet to be proven.

We have defined an error as occurring whenever a method terminates without satisfying its postcondition or whenever it calls another method while violating that method's precondition. To handle the second case and to make our analysis modular, we use a contract rule: a sequent like $\phi \Longrightarrow [\mathrm{foo}()] \psi $ is only provable through the two premises $\phi \Longrightarrow \pre$ and $\phi,\post \Longrightarrow \psi$ where $(\pre,\post)$ is foo's contract.

In addition to potentially slow, but powerful formal verification tools like \KeY, we can also use methods like testing, monitoring or expert estimates.
While such methods may be faster and give a good first approximation, their use in \QAC{} is generally unsound.

\section{Theoretical Overview of \QAC{}}
\label{sec:overview}

\begin{figure*}[tb]
  \centering
  \resizebox{.7\textwidth}{!}{
    \begin{tikzpicture}[inner sep=12pt]
      \node[draw,text width=7em,text centered,very thick]%
      (PCM) {\bfseries Architectural\\Model};

      \node[draw,left=of PCM,text width=4em,text centered,xshift=0.3cm]%
      (K) {Coverage Regions};

      \node[draw,rectangle,left=5 of PCM,text width=6.5em]%
      (J) {Source Code \&\\Contracts};

      \draw[->] ([yshift=3mm]J.east) -- node[yshift=-3mm,above] {Verification} ([yshift=3mm]K.west);

      \draw[->] ([yshift=-3mm]J.east) -- node[yshift=3mm,below] {Test} ([yshift=-3mm]K.west);
      Architecture modeling with Palladio

      \node[draw=none,below=of PCM,text width=5em,text centered] (E) {
        \tikz{\umlactor[scale=2]{User}}
      };

      \draw[->] ([yshift=-3mm]J) |- node[xshift=3mm,yshift=3mm,below] {Review \&
        Estimation} (E);

      \draw[->] (E) -- node[align=right,left,text width=6em]{Modeling\\Architecture} (PCM);

      \node[draw,right=of PCM, text width=5em,text centered] (DS) {Probabilistic\\Model};

      \node[draw,rectangle,below=2 of DS,text width=5em]%
      (R) {Reliability\\Estimation};

      \draw[->] (K) -- (PCM);

      \draw[->] (PCM) -- (DS);
      \draw[->] (DS) -- node[yshift=-3mm,align=right,left,text width=5em]{Model Checker} (R);
    \end{tikzpicture}
  }
  \caption{An overview of the \QAC{} workflow.}
  \label{fig:pipeline}
\end{figure*}
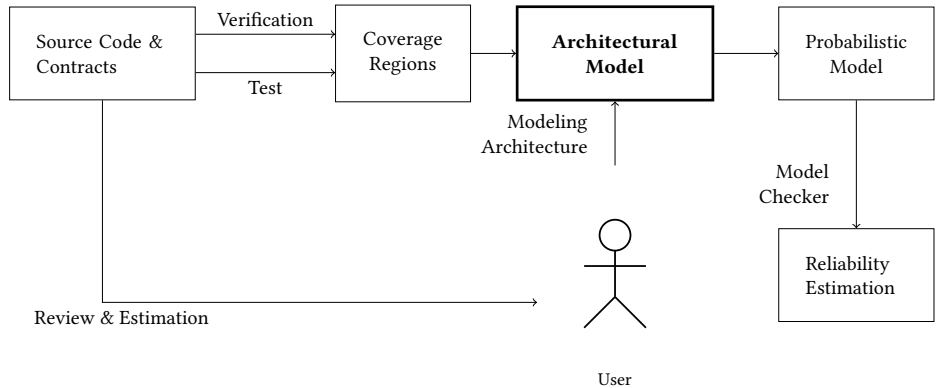

\begin{figure*}
\centering
\resizebox{.7\textwidth}{!}{
\begin{tikzpicture}
\umlclass[x=-4]{WindTurbine}{}{produce(int windSpeed);}
\umlclass[x=0]{Network}{int load;}{addLoad(int n);\\useLoad(int n);}
\umlclass[x=4]{Consumer}{}{consume(int demand);}

\umlassoc[]{WindTurbine}{Network}
\umlassoc[]{Network}{Consumer}

\umlnote[x=-4,y=-3,width=20em]{WindTurbine}{
\texttt{%
void produce(int windSpeed) \{\\
\ \ if (windSpeed < 9)\\
\ \ \ \ debuglog("producing");\\
\ \ \ \ network.addLoad(windSpeed*3/4);\\
\}
}
}
\umlnote[x=4,y=-3,width=15em]{Consumer}{
\texttt{%
void consume(int demand) \{\\
\ \ debuglog("consuming");\\
\ \ network.useLoad(demand);\\
\}
}
}
\end{tikzpicture}
}
  \caption{An example software architecture and implementation.}
  \label{fig:example}
\end{figure*}
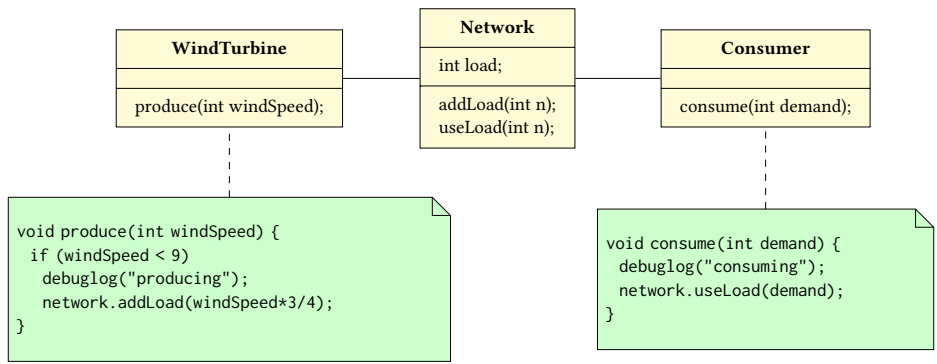

Figure~\ref{fig:pipeline} illustrates our approach. The user-defined
architectural model is in the center, defining the software \emph{components} and \emph{services}.
Every service is modeled by a behavioral specification. In addition, the component model contains a usage profile, which gives a distribution for the initial state and tells us what services are called in the usage of the system.
On the left, we have the source code, which we assume is consistent with the behavioral specifications, and which includes a contract for every method. After calculating coverage regions, we add them to the service models.  
The extended service models and the usage profile are then translated into a probabilistic model, from which a model checker computes the \emph{coverage probability}, i.e., the probability that no service is ever called with an input outside its coverage region. If all coverage regions are correct, the coverage probability is less than or equal to the \emph{correctness probability}, i.e., the probability that the implementation never violates its specification.

\Cref{sec:modeling} defines the parts of our architecture model in more detail.
\Cref{sec:soundness} then explains how we compute the coverage probability from the architecture model on a theoretical level and proves that the coverage probability always under-approximates the correctness probability.
\Cref{sec:key} explains how coverage regions are computed using deductive verification or testing.
\Cref{sec:synthesis} explains how the computation of the coverage probability is implemented.

\newcommand\cov{\mathit{cov}}
\newcommand\body{\mathit{beh}}
\newcommand\Fml{\mathit{Fml}}
\newcommand\Term{\mathit{Term}}
\newcommand\cond{\mathit{cond}}
\newcommand\service{\mathit{s}}

\subsection{Modeling the System}
\label{sec:modeling}

\paragraph{Modeling a service}
We model a service by two parts: its \emph{coverage region}, a formula describing the inputs for which its implementation is proven to behave correctly, and its \emph{behavioral specification}, which is an abstraction of the implementation describing how the service modifies the component state and what other services it calls.
In addition, we assume that every service has a contract ($\pre$,$\post$).
\begin{definition}[Service model]
\label{def:service}
Given a service $s$, let $\Sigma$ be a signature which includes the parameters of $s$ and every component's state variables.
Let $\Fml_\Sigma$ be the set of all quantifier-free first-order formulas and $\Term_\Sigma$ the set of all terms over $\Sigma$.

A service model is a tuple $(\cov, \body)$, where $\cov \in \Fml_\Sigma$ is a coverage region and $\body$ is a sequence of statements: every statement is either a conditional assignment $\mathtt{if}\ (\mathit{fml})\ \mathit{var} = \mathit{trm}$ assigning the value of $\mathit{trm} \in \Term_\Sigma$ to a variable $\mathit{var} \in \Sigma$ of the same type if $\mathit{fml} \in \Fml_\Sigma$ holds, a conditional service call $\mathtt{if}\ (\mathit{fml})\ \mathit{var} = f(\mathit{trm}_1,\ldots,\mathit{trm}_n)$ calling the service $f$ with the well-typed parameters $\mathit{trm}_i$ and assigning the result to $\mathit{var}$ if $\mathit{fml}$ holds, or a conditional termination statement $\mathtt{if}\ (\mathit{fml})\ \lightning$ that terminates the service prematurely.
\end{definition}

\paragraph{Coverage regions}
A coverage region $\cov\in \Fml_\Sigma$ describes the input parameters and component state under which a service's implementation is proven to behave correctly, i.e., not terminate in a state that violates the postcondition nor call another method in a state that violates the callee's precondition.
It is determined by validating the method that implements the service separately from all other methods using testing, formal methods, or run-time monitoring. 
For services which are not yet implemented in an early development stage, the coverage region can also be estimated to obtain a what-if analysis.
A coverage region does not necessarily represent the set of initial states for which the service behaves correctly, but the possibly smaller set of initial states for which we can prove that this is the case.
For the soundness in \Cref{def:systemcorrectness}, we must require that every coverage region be correct, i.e., that no input inside the coverage region lead to an error.
\begin{definition}[Errors and correctness regions]
  \label{def:error}
  For a given method $f$ and its postcondition $\post$, an error is either of the following: a terminal program state for $f$ in which $\post$ does not hold, or a program state in which $f$ calls another method $g$ with precondition $\pre$ in which $\pre$ does not hold.

  We define the method's correctness region $\mathit{correct}$ to be the formula which holds exactly in those initial states in which executing $f$ does not lead to an error.
\end{definition}
\begin{definition}[Correct coverage regions]
  \label{def:correcteregion}
  For a given method $f$ and its contract $(\pre, \post)$, a coverage region $\cov$ is correct
  iff $\cov \rightarrow \mathit{correct}$.
\end{definition}
Since a coverage region describes the inputs for which a service's implementation is proven to behave correctly, any coverage region computed by a sound proof calculus is correct.
While any correct coverage region suffices to make our approach sound, larger regions (i.e., weaker formulas) make it less precise by decreasing the coverage probability. The goal is thus to find maximal correct coverage regions.

\paragraph{Behavioral specifications}
The behavioral specification $\body$ of a service model $\service$ is a sequence of statements according to \Cref{def:service}.

The implementation may be more complex than $\body$ and contain additional statements not represented in $\body$. However, the implementation's behavior w.r.t.\ calls to other services and changes to state variables must be specified exactly. This is because, e.g., a spurious call in a service model may both decrease (if it leads to an error) and increase (if it changes the component state in way that prevents a future error) the coverage probability.

This requirement is achievable with existing methodologies:
In model-driven development, behavioral specifications are created first and then source code is generated from them, being kept consistent with the specification by some consistency preservation approach.
However, one can also go the other way around and infer behavioral specifications from source code.

\paragraph{Example}
From now, we will use the architecture and source code from \Cref{fig:example} as a running example. It implements a heavily simplified simulation of an energy network, consisting of three components that offer four services in total. The $\mathrm{produce}()$ service produces electricity based on the current wind speed; if the wind speed is too high, the service produces no electricity. The $\mathrm{consume}()$ service consumes electricity. And the $\mathrm{addLoad}()$ and $\mathrm{useLoad()}$ services are called by the other two services to modify the current network load.
We assume that $\mathrm{addLoad}()$ and $\mathrm{useLoad}()$ are both specified by the contract $\pre=\post=(\mathrm{load} \geq 0)$ and the other two services by $\pre=(\mathrm{windSpeed} \geq 0),\post=\mathit{true}$ and $\pre=(\mathrm{demand} \geq 0),\post=\mathit{true}$ respectively. In other words, an error occurs if the network load falls below zero.

Since the source code is a refinement of the architecture, it may contain elements not present in the architecture. In our example, the calls to the debug log are present in the source code, but are not a part of the behavioral specifications.

The maximal coverage region for $\mathrm{useLoad}()$ is $n \leq \mathrm{load}$.
For all other services, the maximal coverage region is $\mathit{true}$. This holds even though $\mathrm{consume}()$ may indirectly cause an error by calling $\mathrm{useLoad}()$ because coverage regions only consider errors caused by the service directly.

\subsection{Approximating the Correctness Probability}
\label{sec:soundness}

Using the parts of the architecture model introduced in the previous subsection, we can compute the coverage probability $1 - Pr( \lightning \mid \llbracket \mathcal U(S)\rrbracket)$ that executing the service models $S$ under the usage profile $\mathcal U$ does not lead to a premature termination. If all coverage regions in $S$ are correct, this is less than or equal to the system's correctness probability.

\paragraph{Executable Semantics}
We start by defining what it means to execute a service model $(\cov, \body)$: If $\cov$ does not hold, we immediately terminate prematurely. Otherwise, we execute the behavioral specification $\body$.
Note that according to \Cref{def:service}, every service called by $\body$ must either terminate prematurely or return to the caller.

To execute a set of service models $S$, we must also have a usage profile that tells us what services are called and with what values.
This profile models how the system's implementation is used or expected to be used. It is similar to a service model, but also includes probabilistic variable assignments to model probabilistic user actions. In addition, the usage profile may only call a service if its precondition holds.
\begin{definition}[Usage profile]
    A usage profile $\mathcal U$ is a sequence of statements. Every statement is either a conditional assignment, service call, or termination statement as in \Cref{def:service}, or a conditional probabilistic assignment $\mathtt{if}\ (\mathit{fml})\ \mathit{var} \sim \mathit{dist}$ where $\mathit{dist}$ is a distribution over the domain of the type of $\mathit{var}$.  The condition of every conditional service call in $\mathcal U$ must imply the service's precondition.
\end{definition}

\begin{definition}[Semantics of usage profile]
    Let $\mathcal U$ be a usage profile and $S$ a set of service models. Then  
    $\llbracket \mathcal U(S) \rrbracket$ denotes the set of all finite computation traces which are created by executing $\mathcal U$ and $S$.
\end{definition}
Each trace $t \in \llbracket \mathcal U(S) \rrbracket$ has a probability of occurrence $Pr(t)$ which -- since the service models are deterministic -- is determined solely by the distributions in the probabilistic assignments in $\mathcal U$.
The coverage probability can then be computed as the probability that a trace randomly chosen according to the distributions in the usage profile does not terminate prematurely, i.e., never calls a service outside its coverage region.

Returning to our example from \Cref{fig:example}, we assume a usage profile which calls $\mathrm{produce}(\mathrm{windSpeed})$ followed by $\mathrm{consume}(\mathrm{demand})$ where $\mathrm{windSpeed}$ is an integer uniformly distributed in $[5,9]$ and  $\mathrm{demand}$ is an integer uniformly distributed in $[0,4]$. Then \QAC{} computes a probability of $\frac{4}{5}$, which is equal to the actual correctness probability. This can be seen by considering that of the 25 possible inputs, only the input $\mathrm{windSpeed} = 5, \mathrm{demand} = 4$ and the four inputs $\mathrm{windSpeed} = 9, \mathrm{demand} > 0$ lead to an error.
A smaller coverage region leads to a lower coverage probability. E.g., if we take the coverage region of $\mathrm{useLoad}()$ to be $\mathit{false}$, the probability becomes $0$.

\paragraph{Soundness}
To ensure that the coverage probability un\-der-\allowbreak{}ap\-prox\-i\-mates the correctness probability, it suffices to show that all coverage regions are correct.
\begin{theorem}
\label{def:systemcorrectness}
Let $\mathcal U$ a usage profile.
Let $S, S'$ be sets of service models s.t.\ for each service $\service$ the coverage region in $S$ for $\service$ is smaller or equal to the corresponding one in $S'$. Then 
$
   Pr( \lightning \mid \llbracket \mathcal U(S)\rrbracket)
    \geq Pr(\lightning\mid \llbracket \mathcal U(S')\rrbracket),
$
where $Pr( \lightning \mid \llbracket \mathcal U(S)\rrbracket)$ expresses the probability that by executing the usage profile $\mathcal U$ using service models $S$, we terminate prematurely.
\end{theorem}
\begin{proof}
  The probability $Pr( \lightning \mid \llbracket \mathcal U(S)\rrbracket)$ is the sum of probabilities of every prematurely terminating trace: $\sum_{t\in  \llbracket \mathcal U(S)\rrbracket} Pr( \lightning \mid t)\,Pr(t \mid \llbracket \mathcal U(S)\rrbracket)$.
  Note that the service models are deterministic and the only probabilistic choices are in the usage
  model $\mathcal U$. Thus, there is an isomorphism  $\llbracket \mathcal U(S)\rrbracket \simeq \llbracket \mathcal U(S')\rrbracket$.
  Assume $t \in \llbracket \mathcal U(S)\rrbracket$ and a corresponding $t' \in \llbracket \mathcal U(S')\rrbracket$. 
  For a single $t$, the probability $Pr( \lightning \mid t)$ is either 0 or 1.
  Thus, we need to show that if $t'$ terminates prematurely, then so does $t$. 
  Due to the assumption that each coverage region in $S$ is smaller than or equal to the corresponding one in $S'$, we know that if $t'$ hits calls a service $s$ outside its coverage region, then either the same service $s$ is also called outside its coverage region in $t$, or $t$ already terminates before $s$ is called.
\end{proof}

This theorem tells us that for any service model $S$ consisting only of correct coverage regions, $Pr( \lightning \mid \llbracket \mathcal U(S)\rrbracket)$ over-\allowbreak{}ap\-prox\-i\-mates the actual error probability, and thus $1 - Pr( \lightning \mid \llbracket \mathcal U(S)\rrbracket)$ un\-der-\allowbreak{}ap\-prox\-i\-mates the actual correctness probability of the implementation.

\paragraph{Guiding with quantitative values}
To apply \QAC{}, we must specify and verify all components and all services they provide. 
However, formally specifying and verifying every service is very labor-intensive.
Instead, we can sacrifice soundness for practicability by using \QAC{} to calculate the probability of a certain service being called and ignoring all rarely-called services in our analysis.

Furthermore, the severity of different kinds of errors can be considered:
For example, logging is pervasive; each service may invoke the logging service to trace the data processing. But an error during logging may be uncritical if it does not influence other parts of the program.
To address this, we can manually set the coverage regions of uninteresting services to $\mathit{true}$ or even remove these service from the model entirely (like we did with the debuglog() in our running example) to focus only on those errors which interest us.
Alternatively, we can extend the architecture model with 
a new attribute which reflects the cost of an error in a service and instrument the probabilistic model generated by \QAC{} such that it approximates the expected error cost instead of the correctness probability.

\section{Computation of Coverage Regions}
\label{sec:key}
This section explains how coverage regions can be computed.
As mentioned before, we can use verification tools like \KeY{} or testing approaches, which, while faster, give incorrect results.

\paragraph{Using formal verification to compute coverage regions}

Consider a JML contract $(\pre,\post)$.
Beckert et al.~\cite{KeYCoverage} introduce the notion of \emph{state space coverage}.
Given a partially open proof, this is the set of initial states for which we know that the method will satisfy its postcondition.
Assuming only locally sound proof rules were applied, this set consists of all entry state in which all open goals hold. Thus, the set of open goals induces a coverage region, as formalized by the following theorem and corollary.

\begin{theorem}[Local Contract Satisfaction]
\label{thm:key:local-sat}
Let $\rho = \pre \rightarrow \left[ \service \right] \post$ be a JavaDL proof obligation for some contract.
Let $\left(\mathscr{O},\mathscr{C}\right)$ be the open and closed goals of an unfinished proof produced through locally sound rules%
.
If a state $\sigma$ satisfies all open goals, i.e., $\sigma \vDash \bigwedge_{\left( \phi \Longrightarrow \psi \right) \in \mathscr{O}} \left( \bigwedge\limits_{i=1}^n \phi_i \rightarrow \bigvee\limits_{j=1}^m \psi_j \right)$, then $\sigma \vDash \rho$.
\end{theorem}
\begin{proof}
Remember that for a locally sound proof rule, any state satisfying all premises also satisfies the conclusion. Thus, any state $\sigma$ which satisfies the conjunction of all open and closed goals $\bigwedge_{\left( \phi \Longrightarrow \psi \right) \in \mathscr{O} \cup \mathscr{C}} \left( \bigwedge\limits_{i=1}^n \phi_i \rightarrow \bigvee\limits_{j=1}^m \psi_j \right)$ satisfies $\rho$.
Since all closed goals are (universally) valid, any state $\sigma$ which satisfies the conjunction of all open goals only
also satisfies $\rho$.
\end{proof}
 
This result directly induces a correct coverage region:

\begin{corollary}[Open Branches as Correct Coverage Region]
\label{cor:key:erroneous}
Let $\rho$ and $\left(\mathscr{O},\mathscr{C}\right)$ be as before.
The following formula is an correct coverage region for $\rho$:
\[
\cov \equiv \bigwedge\limits_{\left( \phi \Longrightarrow \psi\right) \in \mathscr{O}}
\left(\bigvee\limits_{i=1}^n \lnot\phi_i \lor \bigvee\limits_{j=1}^m \psi_j\right).
\]
\end{corollary}
Given such a $\cov$, any $\cov'$ such that $\cov' \rightarrow \cov$ is also a correct coverage region.
In particular, for a $\cov$ in conjunctive normal form, we may remove any atoms within a disjunction that contain variables not represented in the architecture model.

\begin{definition}[Projection]
    Let $\cov \equiv \bigwedge_i \bigvee_j A_{i j}$ be a formula in conjunctive normal form. Let $V$ be a set of variables. The formula obtained by removing all $A_{i j}$ with variables not in $V$ is called the projection of $\cov$ on $V$.
\end{definition}

To summarize, for any given service, we use \KeY{} to (automatically or manually) find a partial proof for the service's implementation. We then construct a coverage region from the partial proof's open goals. Lastly, we make that coverage region usable in the architecture model by projecting away all variables that exist in the implementation but not the model.

\paragraph{Approximate results}
\label{sec:approx}

Testing results can be leveraged in two manners:
Either the complement of failed tests shapes a possibly incorrect coverage region,
or
the successful tests shape a correct, but small coverage region.
Similarly, results from the monitoring of deployed software components are usable to obtain approximate coverage regions.

For components which have not yet been implemented, neither
verification nor testing is possible.
In this case, a domain expert must estimate the coverage
regions.
This is useful as an approximation during development. After the service is implemented, the estimate can be replaced by tests or verification results.

\section{From Coverage Regions to Probabilities}
\label{sec:synthesis}

In this section, we describe how the probabilities to enter a coverage region under a given usage profile are calculated. 
In \Cref{sec:integration}, we define the necessary elements of the Palladio Component Model (PCM) and explain how to integrate the computed coverage regions.
In \cref{sec:calcProb}, we explain how to use the resulting model to calculate the coverage probability.

\subsection{Integration of Coverage Regions in the Architectural Model}
\label{sec:integration}

\QAC{} builds on usage profiles and elements of a component-based architecture
as presented in Sections~\ref{sec:prelim} and~\ref{sec:overview}.

Relevant elements of behavioral specifications for \QAC{} are \begin{inparaenum}[(1)]
    \item branching nodes that specify the branching criteria based on the value of a quantifier-free formula over the values of parameters or by a fixed absolute probability.
    \item bounded loop nodes that specify the number of iterations based on the values of parameters.
    \item external call nodes, which define a call to a service and the provided input for the parameters. These calls may contain a condition over the values of parameters as well as results of other services invocations.  
\end{inparaenum}

To support \QAC{}'s service models in the PCM, it must also model the relevant state of a component using state variables with a name, type, and a probability distribution describing their possible initial values. 
We introduce state variables and setter and getter nodes in architectural behavioral specifications for the introduction of coverage regions.
We also extend all relevant elements which contain some kind of condition to handle values of state variables in addition to parameter values.
We extend the PCM with our required information by applying the \emph{inheritance} and \emph{plain referencing} extension approach, described by Heinrich et al.~\cite{heinrich2021}.
We provide a separate metamodel with classes extending the existing meta-classes of the PCM, e.g., the getter and setter accesses, or elements which reference an existing element, such as the assignment of state variables to components.
\begin{figure}
\begin{subfigure}[b]{0.5\textwidth}
\centering
\scalebox{.9}{
\begin{tikzpicture}[box/.style={fill=blue!15, draw, minimum height=1cm, text width=4cm, align=left},scale=1]
\umlstateinitial[name=Init]

\begin{umlstate}[x=4,y=-.4,name=Call1]{Conditional service call}
\umlstatetext{%
if($\mathrm{windSpeed} < 9$)\\
\qquad{}network.addLoad(windSpeed*3/4)
\vspace{-1cm}
}
\end{umlstate}

\umlstatefinal[anchor=west,right=.75 of Call1.east,name=Final]

\umltrans[]{Init}{Call1}
\umltrans[]{Call1}{Final}

\node[box,anchor=north west] (State) at (4.5,3) {%
\emph{Required components:}\\
network: Network\\
\vspace{-.5em}
\hrulefill\\
\emph{State variables:}\\
\ ---
};

\node[box,anchor=north east,text width=2.5cm] at (4.5,3) (Crit) {%
\emph{Coverage region:}\\
$\mathit{true}$
};

\node[left=2pt of Init.west] (Left) {};
\node[below=2pt of Call1.south] (Bottom) {};

\begin{pgfonlayer}{background}
\node[draw, inner ysep=0pt,inner xsep=0pt, fill=blue!15, fit=(Bottom)(Left)(State)] (around) {};
\end{pgfonlayer}
\end{tikzpicture}
}
\caption{Behavioral specification for produce().}
\label{fig:behavior-specification-produce}
\end{subfigure}
\begin{subfigure}[b]{0.5\textwidth}
\centering
\scalebox{.9}{
\begin{tikzpicture}[box/.style={fill=blue!15, draw, minimum height=1cm, text width=4cm, align=left},scale=1]
\umlstateinitial[name=Init]

\begin{umlstate}[x=3.5,y=-.4,name=Call1]{Conditional service call}
\umlstatetext{%
if(true)\\
\qquad{}network.useLoad(demand)
\vspace{-1cm}
}
\end{umlstate}

\umlstatefinal[anchor=west,right=.75 of Call1.east,name=Final]

\umltrans[]{Init}{Call1}
\umltrans[]{Call1}{Final}

\node[box,anchor=north west] (State) at (3.7,3) {%
\emph{Required components:}\\
network: Network\\
\vspace{-.5em}
\hrulefill\\
\emph{State variables:}\\
\ ---
};

\node[box,anchor=north east,text width=2.5cm] at (3.7,3) (Crit) {%
\emph{Coverage region:}\\
$\mathit{true}$
};

\node[left=2pt of Init.west] (Left) {};
\node[below=2pt of Call1.south] (Bottom) {};

\begin{pgfonlayer}{background}
\node[draw, inner ysep=0pt,inner xsep=0pt, fill=blue!15, fit=(Bottom)(Left)(State)] (around) {};
\end{pgfonlayer}
\end{tikzpicture}
}
\caption{Behavioral specification for consume() .}
\label{fig:behavior-specification-consume}
\end{subfigure}
\begin{subfigure}[b]{0.5\textwidth}
\centering
\scalebox{.9}{
\begin{tikzpicture}[box/.style={fill=blue!15, draw, minimum height=1cm, text width=4cm, align=left},scale=1]
\umlstateinitial[name=Init]

\begin{umlstate}[x=3,y=-.35,name=Call1]{Set}
\umlstatetext{%
load = load + n
\vspace{-1cm}
}
\end{umlstate}

\umlstatefinal[anchor=west,right=.75 of Call1.east,name=Final]

\umltrans[]{Init}{Call1}
\umltrans[]{Call1}{Final}

\node[box,anchor=north west] (State) at (3,3) {%
\emph{Required components:}\\
\ ---\\
\vspace{-.5em}
\hrulefill\\
\emph{State variables:}\\
\ int load (initial value: 0)
};

\node[box,anchor=north east,text width=2.5cm] at (3,3) (Crit) {%
\emph{Coverage region:}\\
$\mathit{true}$
};

\node[left=2pt of Init.west] (Left) {};
\node[below=2pt of Call1.south] (Bottom) {};

\begin{pgfonlayer}{background}
\node[draw, inner ysep=0pt,inner xsep=0pt, fill=blue!15, fit=(Bottom)(Left)(State)] (around) {};
\end{pgfonlayer}
\end{tikzpicture}
}
\caption{Behavioral specification for addLoad().}
\label{fig:behavior-specification-addload}
\end{subfigure}
\begin{subfigure}[b]{0.5\textwidth}
\centering
\scalebox{.9}{
\begin{tikzpicture}[box/.style={fill=blue!15, draw, minimum height=1cm, text width=4cm, align=left},scale=1]
\umlstateinitial[name=Init]

\begin{umlstate}[x=3,y=-.35,name=Call1]{Set}
\umlstatetext{%
load = load - n
\vspace{-1cm}
}
\end{umlstate}

\umlstatefinal[anchor=west,right=.75 of Call1.east,name=Final]

\umltrans[]{Init}{Call1}
\umltrans[]{Call1}{Final}

\node[box,anchor=north west] (State) at (3,3) {%
\emph{Required components:}\\
\ ---\\
\vspace{-.5em}
\hrulefill\\
\emph{State variables:}\\
\ int load (initial value: 0)
};

\node[box,anchor=north east,text width=2.5cm] at (3,3) (Crit) {%
\emph{Coverage region:}\\
$n \leq \mathrm{load}$
};

\node[left=2pt of Init.west] (Left) {};
\node[below=2pt of Call1.south] (Bottom) {};

\begin{pgfonlayer}{background}
\node[draw, inner ysep=0pt,inner xsep=0pt, fill=blue!15, fit=(Bottom)(Left)(State)] (around) {};
\end{pgfonlayer}
\end{tikzpicture}
}
\caption{Behavioral specification for useLoad().}
\label{fig:behavior-specification-useload}
\end{subfigure}
\caption{Behavioral specifications for the services from \Cref{fig:example}.}
\label{fig:behavior-specifications}
\end{figure}
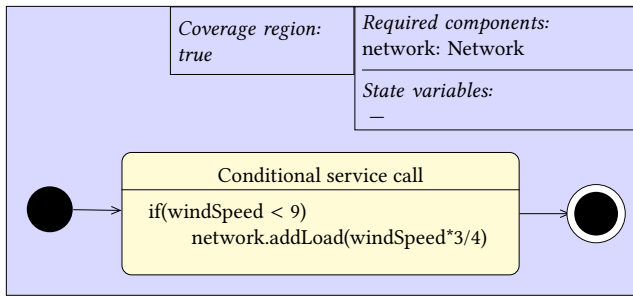
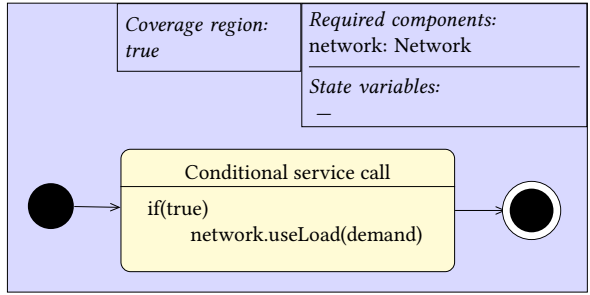
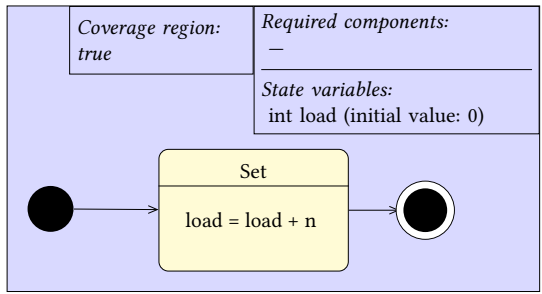
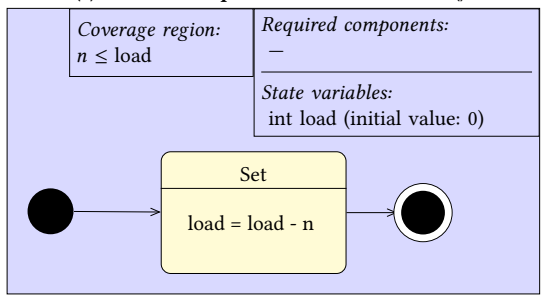
With this extension, we can insert the coverage regions as a projection on the available state variables and parameters in the model into the behavioral specification of a service. 
Conditions for service calls and the symbolic values of the actual arguments can be modeled by developers as part of the behavioral specifications, or computed by \KeY{} in a similar way as the coverage regions. 
The rest of the behavioral specification can be manually modeled or extracted from the source code with reverse engineering approaches, such as the one by Becker et al.~\cite{becker10}.

Figure~\ref{fig:behavior-specifications} depicts the behavioral specifications corresponding to the services from our running example.

\subsection{Calculation of the Probabilities}
\label{sec:calcProb}
Given a usage profile and the architectural model with the information described in \Cref{sec:integration}, we compute the coverage probability.
To this end, we translate each service model into a function in a probabilistic program based on the service model's executable semantics as presented in \Cref{sec:overview}. The usage profile is translated into the main function with which the execution starts.
Calls to services by the usage profile or by other services can be translated into function calls.
Branching nodes and bounded loops are translated to \lstinline|if| or \lstinline|while| statements in the functions with a corresponding condition.

\begin{figure}[t]
\begin{lstlisting}[language=Java,mathescape=true]
int load;

fun main():
  load = 0; windSpeed ~ U(5,9); demand ~ U(0,4);
  if (windSpeed >= 0): produce(windSpeed));
  if (demand >= 0): consume(demand);

fun produce(int windSpeed):
  if (!(true)): $\lightning$;
  if (windSpeed < 9): addLoad(windSpeed * 3 / 4);
  
fun consume(int demand):
  if (!(true)): $\lightning$;
  if (true): useLoad(demand);

fun addLoad(int n):
  if (!(true)): $\lightning$;
  load = load + n;

fun useLoad(int n):
  if (!(n <= load)): $\lightning$;
  load = load - n;
\end{lstlisting}
\caption{Probabilistic program corresponding to \Cref{fig:example}.}
\label{fig:prob-program}
\end{figure}

For example, let us again consider the behavioral specifications in Figure~\ref{fig:behavior-specifications}.
We assume the same usage profile given in \Cref{sec:overview}.
Then the translation into a probabilistic program is shown in \Cref{fig:prob-program}.
Now, a model checker or model counter, e.g., PSI~\cite{PSIsolver}, Storm~\cite{hensel-2022}, or counterSharp~\cite{teuber-2021}, can resolve the probabilities in \Cref{fig:prob-program}.
 
\subsection{Prototypical implementation}

We provide a prototypical implementation of \QAC{} using \KeY{} as the verification tool and the PCM as an architectural description language~\cite{florian_lanzinger_2022_7473812}.
We extend the PCM by introducing state variables, getter and setter actions in the SEFFs as well as elements for the attributes of our service model~(\Cref{def:service}). 

We extract the coverage regions with \KeY{} and project the results onto the state variables and parameters present in the model.
We can also use \KeY{} to extract conditions of service calls.

We implement a transformation from the extended PCM into a probabilistic program in the PSI~\cite{PSIsolver} language similar to the one described in \Cref{sec:calcProb}. Transformations to other model checkers are possible, but not yet implemented.

\section{Case Study}
\label{sec:case-study}
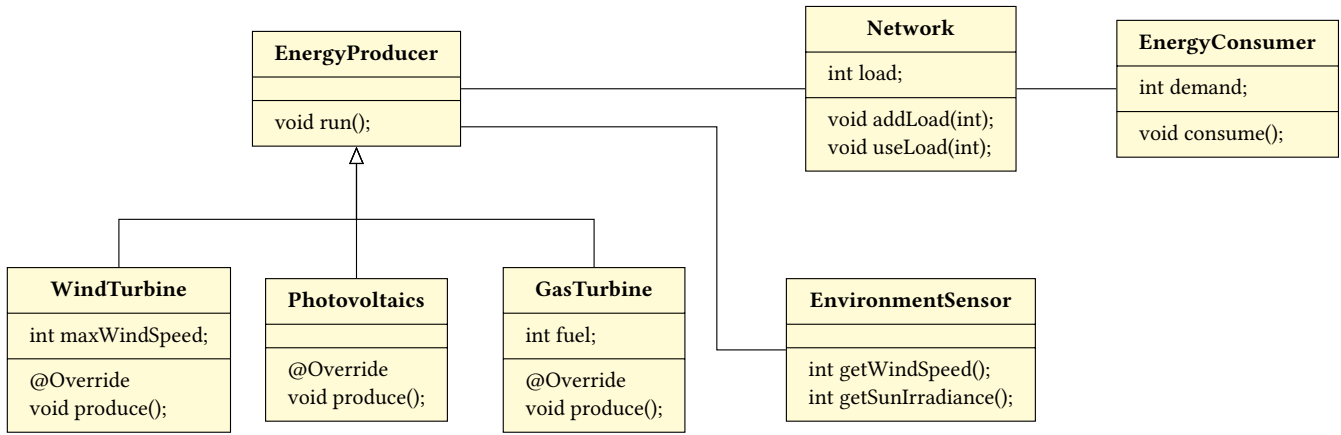
\begin{figure*}[tb]
\centering
\resizebox{1\textwidth}{!}{%
\begin{tikzpicture}[scale=1]
\umlclass[x=0]{EnergyProducer}{}{
  void run();
}
\umlclass[x=-3,y=-3.3]{WindTurbine}{int maxWindSpeed;}{@Override\\void produce();}
\umlclass[x=0,y=-3.3]{Photovoltaics}{}{@Override\\void produce();}
\umlclass[x=3,y=-3.3]{GasTurbine}{int fuel;}{@Override\\void produce();}

\umlinherit[geometry=|-|]{WindTurbine}{EnergyProducer}
\umlinherit[geometry=|-|]{Photovoltaics}{EnergyProducer}
\umlinherit[geometry=|-|]{GasTurbine}{EnergyProducer}

\umlclass[x=7]{Network}{
  int load;
}{
  void addLoad(int); \\
  void useLoad(int);
}

\umlclass[x=11]{EnergyConsumer}{
  int demand;
}{
  void consume();
}

\umlclass[x=7,y=-3.3]{EnvironmentSensor}{}{
  int getWindSpeed();\\
  int getSunIrradiance();
}

\umlassoc[geometry=--]{EnergyProducer}{Network}
\umlassoc[geometry=--]{EnergyConsumer}{Network}
\umlassoc[geometry=-|-,weight=0.65,anchor1=-20]{EnergyProducer}{EnvironmentSensor}
\end{tikzpicture}
}
\caption{Class diagram for the case study.}
\label{fig:casestudy-uml}
\end{figure*}

\begin{figure}[t]
\begin{lstlisting}[language=Java,mathescape=true]
wt.maxWindSpeed = 15; gt.fuel = 2;

repeat ITERATIONS times:
  cons.demand = $D_1$;
  sens.sunIrradiance = $D_1$;
  sens.windSpeed = $D_2$;
  
  wt.produce(); phv.produce();
  if netw.load < cons.demand:
    gt.produce();
  cons.consume();
\end{lstlisting}
\caption{Usage profile for the case study.}
\label{fig:casestudy-usage-profile}
\end{figure}
\paragraph{Structure}
To demonstrate the feasibility of \QAC, we first ran our running example through our implementation and analyzed the resulting probabilistic program with PSI. We obtained the expected result of $\frac{4}{5}$ (see \Cref{sec:overview}). \KeY{} found all coverage regions automatically in 26.3s and PSI computed the probability in 1.3s.\footnote{On average over 3 trials in our artifact VM with 16GB RAM and 1 CPU.}

To test \QAC{}'s scalability, we consider a more complex version of the running example, shown in  \Cref{fig:casestudy-uml}.
Instead of one, we have three energy producers: a wind turbine, a photovoltaic system, and a gas turbine.
In addition, we have an environment sensor, which supplies the current wind speed and sun irradiance.

In our usage profile, shown in \Cref{fig:casestudy-usage-profile}, initial values for the demand, sun irradiance, and wind speed are independently chosen from the random distributions $D_1, D_2$.
Then, both the wt.produce() and phv.produce() services are called and increase the network load based on the current wind speed and sun irradiance respectively. If the wind speed exceeds $15$, the turbine produces no energy.
If the demand is greater than the network load, GasTurbine.produce() may be called as well; this decrements the gas turbine's fuel and increases the network load if there is still fuel left.
Lastly, Consumer.consume() is called and lowers the network load.
All of this is repeated $n$ times. Residual load is carried over into the next iteration.

We consider the following versions of this usage profile: First, we vary the distributions used. In the first version, $D_1 = \mathscr{U}(0,1999), D_2 = \mathscr{U}(0,19)$. In the second version, $D_1 = \mathscr{N}_d(\mu=1000,\sigma=32), D_2 = \mathscr{N}_d(\mu=10,\sigma=3)$ where $\mathscr{N}_d$ is a discretized version of a normal distribution in which all possible values are integers.
For both distributions, we also vary the number of iterations. We consider between 2 to 8 iterations and expect the run time of the model checkers to rise and the coverage probability to sink with the number of iterations.

The system architecture is modeled in the PCM, and there is a SEFF given for every service.
The SEFFs offer a behavioral specification for each service, as explained in \Cref{sec:synthesis}.
In addition, all classes are implemented in Java and specified in JML. Every method implementation is consistent with the corresponding service's SEFF.
The JML specification states that the network load is always non-negative.

\paragraph{Extraction of coverage regions}
We extract the coverage regions (see \Cref{sec:key}) and service call conditions (see \Cref{sec:integration}) from the source code using \KeY{}, then add them to the SEFFs. This yields the coverage region $\mathrm{arg} \leq \mathrm{load}$ for Network.useLoad() and the coverage region $\mathit{true}$ for all other services.%
\footnote{Actually, the coverage regions extracted by \KeY{} have additional conditions that deal with non-nullness, exception freedom, etc., but since these do not interest us, we manually remove them before adding the coverage regions to the SEFFs.}

\paragraph{Computation of the coverage probability}
Having computed the coverage regions, we add them to the SEFFs.
Then, the usage profile and the extended SEFFs are translated into one probabilistic program, as explained in \Cref{sec:calcProb}.

In addition, we manually translated the extended PCM into Prism and checked the Prism program with the Storm model checker.
We also manually translated the extended PCM into a C program that can be checked using the counterSharp approach, which uses model counting in combination with a bounded model checker to quantify the reliability of C programs as the ratio of failing runs to total runs.
As this assumes the initial values to be uniformly distributed, the C program manually converts between a uniform distribution and the distributions shown in our usage profile (\Cref{fig:casestudy-usage-profile}). Unlike PSI and Storm, counterSharp uses approximate model counting, meaning that the result is not necessarily exact, but can be found in much less time.

\paragraph{Results}
\KeY{} ran in 269.3s.\footnote{On average after 3 trials in our artifact VM.}
All instances of our experiments with exact model counters PSI and Storm ran out of memory. This is because creating branches for every conditional service call and coverage region resulted in a very large state space, which is further enlarged by our usage profile containing $3n$ random variables with $(2000 \cdot 2000 \cdot 20)^n$ total values for $n$ loop iterations in the usage profile.

However, using approximate model counting as in the counterSharp approach proved to be practical. As we see in \Cref{tab:results}, although counterSharp shows an exponential run-time increase depending on the number of iterations, its run times are on the low side considering the size of our state space.

We conclude that our approach is feasible and that while exact results can only be achieved for very small programs (like our running example), approximate results with a confidence interval~\cite{Chakraborty-ApproxMC,teuber-2021} can be found in an acceptable run time even for programs with very large state spaces.
The reliability of the intervals can be improved by repeating the computation multiple times.

\newcolumntype{Y}{>{\centering\arraybackslash}X}
\begin{table}[]
    \centering
    \begin{tabularx}{.49\textwidth}{c *{4}{Y}}\toprule
        \multirow{2}{*}{Cycles} & \multicolumn{2}{c}{Run time} & \multicolumn{2}{c}{Coverage probability} \\
        \cmidrule(lr){2-3} \cmidrule(l){4-5}
         &
         $\mathcal{U}$ & $\mathcal{N}_d$ &
         $\mathcal{U}$ & $\mathcal{N}_d$ \\ \midrule
            2 &
             $1.7\sec$ & $6.9\sec$ & 
             $100\%$ & $100\%$ \\ 
            3 &
             $13.5\sec$ & $10.6\sec$ & 
             $\left[99.89\%,99.97\%\right]$ & $100\%$\\ 
            4 &
             $52.5\sec$ & $14.6\sec$ & 
             $\left[99.64\%,99.89\%\right]$ & $100\%$ \\ 
         $\cdots$ & \multicolumn{2}{c}{$\cdots$} & \multicolumn{2}{c}{$\cdots$} \\
            7 &
             $267.8\sec$ & $96.4\sec$ & 
             $\left[97.66\%,99.28\%\right]$ & $100\%$ \\ 
            8 &
             $313.9\sec$ & timeout & 
             $\left[96.58\%,98.95\%\right]$ & timeout \\ 
         \bottomrule
    \end{tabularx}
    \caption{Results of our case study with counterSharp: The actual coverage probability lies in the provided interval with 80\% probability.}
    \label{tab:results}
\end{table}

\section{Related Work}
\label{sec:related-work}

\paragraph{Quantitative and probabilistic program analysis}
Geldenhuys et al.~\cite{geldenhuys-2012} introduce \emph{probabilistic symbolic execution}, where for every branch during symbolic execution, the branch's probability is calculated via model counting and multiplied with the path condition so far. Thus, one obtains a probability for every path.
\QAC{} is also connected to the field of statistical model checking, which applies statistical methods like Monte-Carlo simulation or hypothesis tests to probabilistic models~\cite{YounesKNP06}.
Gerrard et al.~\cite{gerrard-2020} combine over- and under-approximating verification tools to approximate the reachability condition of a given state, resulting in both lower and upper bounds.
\QAC{} differs from these because it is model-based -- i.e., the architectural model we use for the probabilistic analysis is more abstract than the source code -- and modular -- i.e., we can use different analyses for different services.

\paragraph{Model-based Safety and Reliability Analyses}
Different approaches use design models to analyze the safety and reliability of software systems. 
There are approaches that analyze based on design models and their transformations into other formalisms.
Huszerl et al.~\cite{huszerl-2002} transform UML state charts into stochastic reward nets (SRNs) to perform a quantitative dependability analysis. 
Cortelessa et al~\cite{cortelessa-2020} transform architectural UML models into non-functional models in form of General Stochastic Petri Nets (GSPNs).
The generated GSPN is used as an input model for a  safety and reliability analysis.
The approach by Brosch et al.~\cite{brosch-12} performs a reliability analysis based on an extension of the PCM which enables the annotation of absolute probabilities for software and hardware elements.
These approaches rely on the manual annotation of the reliability information retrieved through estimates, experience or calculation.
In contrast, \QAC{} retrieves reliability-relevant information (the critical regions) by an analysis of source code to compute the probabilities to enter critical regions in the software. This approach connects the architectural view explicitly to the final system and avoids erroneous results due to estimations and inconsistencies between model and implementation. 

Work by Töberg et al.~\cite{toeberg22} as well as Tuma et al~\cite{tuma2023checking} is similar to \QAC{} in that it combines architecture and code by using source-code analyses to verify whether assumptions in an architectural model hold in the implementation. However, these approaches do not modify the architectural model based on the results of the source-code analysis.

Kordon et al.~\cite{kordon-2008} propose a discipline they call \emph{verification-driven engineering}, which is based on using software models for verification. They identify several requirements, e.g., a mapping between model elements and a mathematical framework, which we achieve by mapping Palladio components to JML contracts. They also identify several challenges, like the fact that different model properties must be verified with different tools that are difficult to unify in one model.
\QAC{} offers a partial answer by being able to combine different analyses so long as all of their results can be encoded as a critical region.

\paragraph{Cooperative verification}
\QAC{} is related to the field of cooperative verification. Many approaches exist~\cite{ahrendt-2016,DBLP:conf/tacas/BeyerK22,DBLP:conf/birthday/BosH22} that allow tools to share proof obligations or partial results so that different components of a program or different properties of the same component can be proven by different tools. Alternatively, partial verification results can be used to guide test cases.
\QAC{}'s coverage regions serve both as an exchange format that allows multiple approaches to used in \QAC{} and as partial results that may guide tests or monitoring. But in addition to allowing tools to cooperate, \QAC{} computes a correctness measure to quantify how much of the program has been verified by a given combination of tools.

\section{Conclusion}
\label{sec:conclusion}

\paragraph{Summary}
\QAC{} is a modular approach to quantify software correctness
by combining architecture modeling and formal verification. 
We defined \QAC{} formally and implemented it for a subset of Java programs
using the modeling tool Palladio,
the formal verification tool \KeY{}, and various model counters.
We use \KeY{} 
to analyze each individual service and
extract its critical regions, i.e., the conditions under which it
possibly behaves erroneously.
These regions are combined with the Palladio Component Model, which includes an architecture model and usage profile, to compute an over-approximate error probability for the overall system.
Our case study demonstrates that our approach is feasible. While computing an exact error probability is impractical except for small programs, approximative approaches lead to good run times. 

\paragraph{Future work}
Our approach allows us to modularly combine different kinds of analyses, but we mainly focused on deductive verification.
We would like to investigate other analyses like tests and type checkers and formally establish how they can be used to compute coverage regions.
In addition, we want to investigate how we may increase \QAC{}'s precision with tools that over- instead of under-approximate the correctness probability.
Furthermore, this paper only considered reliability properties. We want to apply \QAC{} to 
security properties by extending the usage profile to capture attacker capabilities and attack costs.

\section*{Data Availability Statement}

Our implementation, along with the input and logs for the case study, is available on Zenodo.~\cite{florian_lanzinger_2022_7473812}

\begin{acks}
We would like to thank the reviewers for their valuable feedback.

This work was supported by funding from the topic
    Engineering Secure Systems and the pilot program Core Informatics of the Helmholtz Association (HGF), KASTEL
    Security Research Labs,
    the research project SofDCar (19S21002) funded by the German Federal Ministry for Economic Affairs and Climate Action, and the German Research Foundation (DFG) under project number 499241390, HE8596/3-1 (FeCoMASS).
\end{acks}

\bibliographystyle{ACM-Reference-Format}
\bibliography{main}

\end{document}